\documentclass[reqno]{amsart}
\usepackage[T1]{fontenc}
\usepackage{lmodern,amsmath,amssymb,amsthm,mathtools,microtype}
\usepackage[margin=0.82in]{geometry}
\usepackage[hidelinks]{hyperref}
\usepackage{algorithm}
\usepackage{algpseudocode}

\newtheorem{theorem}{Theorem}[section]
\newtheorem{proposition}[theorem]{Proposition}
\newtheorem{lemma}[theorem]{Lemma}
\newtheorem{corollary}[theorem]{Corollary}
\newtheorem{remark}[theorem]{Remark}

\theoremstyle{definition}
\newtheorem{definition}[theorem]{Definition}

\DeclareMathOperator{\disc}{disc}
\newcommand{\R}{\mathbb R}
\newcommand{\E}{\mathbb E}
\newcommand{\Prob}{\mathbb P}
\newcommand{\norm}[1]{\lVert#1\rVert}
\newcommand{\ip}[2]{\langle#1,#2\rangle}
\newcommand{\lstar}{\log^*}

\allowdisplaybreaks

\title{Improved Algorithms for Beck--Fiala with Bounded Sets}
\author{Dylan J. Altschuler}
\address{Dylan J. Altschuler, Department of Mathematics, University of Texas at Austin.}
\email{dylan.altschuler@austin.utexas.edu}

\date{}

\begin{document}
\maketitle
\vspace{-1.2em}

\begin{abstract}
We give an efficient algorithm with improved algorithmic guarantees for the (offline) Beck--Fiala problem when the sets have bounded size.  Let $A$ be an arbitrary matrix $A\in\{0,1\}^{m\times n}$ with at most $d$ ones per column and at most $s$ ones per row. Let $\log^*$ denote the iterated logarithm and $\ell_j$ denote the $j$-fold composition of log. Assume $s\le\exp(O(\sqrt d))$. We provide an efficient algorithm that, for arbitrary sparsity $d$, gives $O(\sqrt d(1+\log^*n))$ discrepancy. Moreover, if $d\ge\ell_j(n)$ for a fixed integer $j\ge1$, the algorithm gives $O_j(\sqrt d)$ discrepancy.  The proof is a bootstrapping scheme using the Bansal-Jiang algorithm.
\end{abstract}

\section{Introduction}

For a matrix $A\in\R^{m\times n}$, write
\[ \disc(A)=\min_{\chi\in\{-1,1\}^n}\norm{A\chi}_\infty. \]
If $A$ is the incidence matrix of a set system and each element belongs to at most $d$ sets, the Beck--Fiala conjecture predicted $\disc(A)=O(\sqrt d)$~\cite{BF}. Note that bounds on the maximum row-sum of $A$, which are not assumed in the original Beck-Fiala setting, correspond to bounds on the maximum set size. We consider the task of obtaining efficient (polynomial in the input size) algorithmic discrepancy guarantees for systems with maximum set size of $s$.

\begin{remark}
During the final stages of the writing of this article, Guo, Fang, and Lu used an automatic AI system to prove the stronger Koml\'os conjecture~\cite{GFL}.  Applied to $A/\sqrt d$, this resolves Beck--Fiala with no restriction on the set sizes.  Their proof is existential; the sharp algorithmic problem remains open and is the focus of this article.
\end{remark}

Building on their earlier work~\cite{BJearly}, Bansal and Jiang~\cite{BJ} gave a polynomial-time algorithm for the Beck--Fiala setting that attains discrepancy
\begin{equation}\label{eq:BJ} O\!\left(\min\{d,\sqrt d+(d\log(2n))^{1/3}\}\right), \end{equation}
and hence $O(\sqrt d)$ for $d\gtrsim\log^2n$~\cite[Theorem~5.1]{BJ}. Altschuler and Tikhomirov subsequently improved this to $d\gtrsim_\eta(\log n)(\log\log n)^{2+\eta}$ with an independent approach via an \textit{online} algorithm~\cite{AT}; see also \cite{AT2} which characterizes online discrepancy in the \textit{random} Beck--Fiala setting.  In a complementary direction, if every row has at most $s$ nonzero entries, Harris and Srinivasan give a randomized polynomial-time bound $O(\sqrt{s\log(2d)})$~\cite[Theorem~5.2]{HS}.

We obtain the existential scale algorithmically in bounded-set regimes far below the preceding sparsity thresholds. Define the $j$th smoothed iterated logarithm inductively by
\begin{equation}\label{eq:logs} \ell_0(n)=n, \qquad \ell_{j+1}(n)=\log(e+\ell_j(n)), \end{equation}
and define
\[ \lstar n:=\min\{j\ge0:\ell_j(n)\le e\}. \]
This convention differs from the usual iterated logarithm by at most an additive absolute constant.

\begin{theorem}[Improved algorithmic discrepancy for Beck--Fiala]\label{thm:main}
Let $A\in[-1,1]^{m\times n}$ have at most $d\ge1$ nonzero entries in each column and at most $s\ge1$ in each row.  For every $z\in[-1,1]^n$ and every fixed integer $k\ge1$, there is $\chi\in\{-1,1\}^n$, produced by an efficient randomized algorithm, such that
\begin{equation}\label{eq:fixed} \norm{A(\chi-z)}_\infty \le C_k\min\!\left\{d,\sqrt d+ \bigl[d\{\log(2s)+\ell_{k+1}(n)\}\bigr]^{1/3}\right\}. \end{equation}
There is also a universal constant $C$ for which
\begin{equation}\label{eq:uniform} \norm{A(\chi-z)}_\infty \le C\min\!\left\{d, \bigl[\sqrt d+\{d\log(2s)\}^{1/3}\bigr](1+\lstar n)\right\}. \end{equation}
In particular, taking $z=0$ bounds discrepancy.
\end{theorem}

\begin{remark}[Extension to hereditary discrepancy]
The \textit{hereditary discrepancy} of a matrix is the largest discrepancy of any sub-matrix obtained by considering a subset of the columns. Applying the above theorem to every subset of columns also gives the same bounds for hereditary discrepancy.
\end{remark}

\begin{corollary}[User-friendly formulation]\label{cor:sqrt}
Fix $a>0$ and an integer $j\ge1$. Let $A\in[-1,1]^{m\times n}$ have at most $d\ge1$ nonzero entries per column and at most $\exp(a\sqrt d)$ per row. Then we obtain algorithmic (hereditary) discrepancy of $O_a(\min\{d,\sqrt d(1+\lstar n)\})$. Moreover, under the further assumption that the sparsity is mildly lower bounded, $d\ge\ell_j(n)$, we obtain algorithmic discrepancy of $O_{a,j}(\sqrt d)$.
\end{corollary}
\begin{proof}
Apply~\eqref{eq:uniform} with $s=\exp(a\sqrt d)$, for which $\log(2s)=O_a(\sqrt d)$. If also $d\ge\ell_j(n)$, then $\ell_{j+1}(n)\le\log(e+d)=O(\sqrt d)$, so the second claim follows from~\eqref{eq:fixed} with $k=j$.
\end{proof}

Let us briefly overview the ideas. In~\cite{BJ}, Bansal and Jiang develop a random walk in the cube $[-1,1]^n$. Coordinates are frozen as they approach signs, and linear constraints control the increase of row discrepancies. The main obstacle is to keep enough directions available for the walk to continue making progress toward a signing. The analysis of Bansal and Jiang controls both the coordinate increments and linear combinations of the increments of different rows. The latter control bounds the probability that many rows simultaneously accumulate large discrepancy.

Our modification produces a partial rounding $x$ from an arbitrary starting point $z$, with $\|A(x-z)\|_\infty=O(b)$ and all coordinates outside an exceptional set of columns $B$ rounded to signs. Proposition~\ref{prop:partial} bounds the probability that any specified collection of columns with disjoint row supports all belongs to $B$. The row-size bound then lets us decompose the submatrix obtained by restricting to $B$ into small blocks that can be rounded separately without adding their errors. Repeating this decomposition yields the iterated logarithm bounds. \\

\textbf{AI disclosure.} ChatGPT substantially aided in the proof of Proposition \ref{prop:partial}, specifically in the technical aspects of adapting the analysis in \cite{BJ}.

\section{Bootstrapping scheme}

\begin{definition}\label{def:columngraph}
The \emph{column graph} $G = ([n],E)$ of an $m\times n$ matrix $A$ has a vertex for each column of $A$, and joins two vertices when there is a row that is supported on both corresponding columns. In particular, independent sets of the column graph are collections of columns with disjoint row supports.
\end{definition}

The following decoupling estimate is proved in Section~\ref{sec:joint}; importantly, it imposes no row-size bound.

\begin{proposition}[Joint partial rounding]\label{prop:partial}
There are universal constants $c,C,c_0,d_0>0$ with the following property. Suppose $A\in[-1,1]^{m\times n}$ has absolute column sums at most $d$, $z\in[-1,1]^n$, as well as $d\ge d_0$ and $C\sqrt d\le b\le c_0d$. Then, there is an efficient algorithm computing $x\in[-1,1]^n$ and $B\subseteq[n]$ such that
\begin{equation}\label{eq:partial-error} \norm{A(x-z)}_\infty\le Cb, \qquad x_j\in\{-1,1\} \quad \forall\, j\notin B, \end{equation}
and, for every independent set $R$ in the column graph,
\begin{equation}\label{eq:joint-tail} \Prob(R\subseteq B)\le\exp\{-cb^3|R|/d\}. \end{equation}
\end{proposition}

We now convert the joint tail into a bound on the sizes of the new discrepancy instances that arise in the bootstrapping. This is the sole place that the row-size bounds are needed.

\begin{lemma}[Small subproblems]\label{lem:components}
Let $A$ have $N$ columns and suppose its column graph $G:=G(A)$ has maximum degree at most $\Delta\ge1$. Let $B\subseteq[N]$ be random, and write $G[B]$ for the subgraph induced by $B$. Suppose that, for some $K>0$ and every independent set $R$ of $G$,
\[ \Prob(R\subseteq B)\le e^{-K|R|}. \]
Then, if $K\ge2\log(4\Delta^2)$, there is some constant $C>0$ for which it holds with probability at least $3/4$ that every connected component of $G[B]$ has size at most
\begin{equation}\label{eq:component-size} C(\Delta+1)\left(1+\frac{\log(4N)}K\right). \end{equation}
\end{lemma}

The idea is to associate with each possible connected component of $G[B]$ a single large independent set $R$ of $G$ that has a specific tree structure and must be a subset of $B$. Union bounding over all such trees will give the result.

\begin{proof}
Suppose a connected component $S$ of $G[B]$ has more than $(\Delta+1)(r-1)$ vertices. Starting from any vertex of $S$, select vertices one at a time from $S$, each at $G[B]$-distance two from the previously selected set. If fewer than $r$ vertices have been selected, their closed neighborhoods contain at most $(\Delta+1)(r-1)$ vertices. Hence some vertex of $S$ is outside those neighborhoods. Connectivity of $S$ gives a vertex at distance exactly two from the selected set, so the process can continue until $r$ vertices have been selected.

The selected set $R$ is an independent set. We use $R$ to construct a rooted tree: start from the first vertex added to $R$ and link each subsequent vertex to an earlier vertex at distance two. Ordering the children at each vertex gives a rooted plane tree. The number of such tree shapes with $r$ vertices is the Catalan number $\frac1r\binom{2r-2}{r-1}\le4^{r-1}$~\cite{StanleyEC2}. There are $N$ choices of label for the root, and at most $\Delta^2$ choices for the label of each child vertex. Thus there are at most $N(4\Delta^2)^{r-1}$ possible sets $R$. Each is contained in $B$ with probability at most $e^{-Kr}$, so a union bound gives for every integer $r\ge1$,
\begin{equation}\label{eq:component-probability} \Prob\bigl(G[B]\text{ has a connected component of size } >(\Delta+1)(r-1)\bigr) \le N(4\Delta^2)^{r-1}e^{-Kr}. \end{equation}
In particular, under the stated condition on $K$, the right-hand side is at most $Ne^{-Kr/2}$. Taking $r=\lceil2\log(4N)/K\rceil$ makes this at most $1/4$ and gives \eqref{eq:component-size}.
\end{proof}

Finally, we record the original linear-algebraic bound of Beck and Fiala; while such language was not used explicitly in their original paper, it is immediate from their proof that the result is algorithmic.

\begin{proposition}[Beck--Fiala~\cite{BF}]
Let $F\in \R^{m\times q}$, with each column bounded in $\ell_1$ by $d$. Then, for every $z\in[-1,1]^q$,
\begin{equation}\label{eq:linear} \min_{\chi\in\{-1,1\}^q} \norm{F(\chi-z)}_\infty \le 2d. \end{equation}
Moreover, this is achieved with an efficient algorithm.
\end{proposition}

We are ready to prove the main theorem.

\begin{proof}[Proof of Theorem~\ref{thm:main}]
We repeatedly apply the partial-rounding algorithm, decompose the remaining columns into blocks, and continue rounding each block from its current fractional point. We first explain this decomposition, then bound the block sizes, and finally choose when to stop recursing and complete the rounding.

Consider a current subproblem: a column submatrix $F$ of $A$ with $N$ columns and a starting point $y\in[-1,1]^N$. Initially, $F=A$, $N=n$, and $y=z$. For an admissible parameter $b$, Proposition~\ref{prop:partial} produces $x$ and $B$ with
\[ \norm{F(x-y)}_\infty\le Cb, \qquad x_j\in\{-1,1\}\quad(j\notin B). \]
Keep the signs outside $B$, and remove any already integral coordinates from $B$. If $B$ is empty, this subproblem is complete. Otherwise, let the blocks be the connected components $S$ of $G[B] := G(F)[B]$, the column graph of $F$ restricted to $B$. For each block, we must round $F_S$, the submatrix obtained by restricting $F$ to the columns in $S$, starting from $x_S$.

By construction, no row has support in two different blocks. Consequently, completing these roundings gives
\[ \norm{F(\chi-y)}_\infty \le Cb+\max_S \norm{F_{S}(\chi_S-x_S)}_\infty. \]
Crucially, the errors from different blocks at the same level do not add. Thus each recursive level of the argument adds at most $Cb$ to the discrepancy. Finally, as every block trivially inherits the original bounds $d$ and $s$, we can apply the same procedure again.

We now choose $b$ to ensure that the blocks are small. Take:
\begin{equation}\label{eq:b-choice} b:=C\left(\sqrt d+\{d\log(2ds)\}^{1/3}\right). \end{equation}
Since $(d\log(2d))^{1/3}=O(\sqrt d)$, this satisfies $b=O(\sqrt d+\{d\log(2s)\}^{1/3})$. If $d<d_0$ or $b>c_0d$, the bound $2d$ from \eqref{eq:linear} already proves both conclusions. Otherwise, Proposition~\ref{prop:partial} applies.

The column graph of a subproblem has maximum degree at most $d(s-1)$. With $\Delta=\max\{1,d(s-1)\}$, our choice of $b$ ensures
\[ K:=cb^3/d\ge2\log(4\Delta^2). \]
Lemma~\ref{lem:components} therefore guarantees, with probability at least $3/4$, that every resulting block has at most
\begin{equation}\label{eq:g} g(N):=\left\lceil Cds\log(eN)\right\rceil \end{equation}
columns. If this test fails, discard the output and rerun the partial-rounding algorithm on the same subproblem $F$, from the same starting point $y$, with fresh randomness.

To track the recursion, let $N_0=n$ and $N_{i+1}=g(N_i)$. After $r$ levels of recursion, every unfinished block has at most $N_r$ columns, and the accumulated discrepancy in each row is at most $Crb$. The following simple estimate will handle both conclusions. Choose $a=Cds\ge8$ large enough that $g(x)\le a(1+\log x)$. For $u\ge1$,
\[ g(a^2u) \le a(1+2\log a+\log u) \le a^2\log(e+u). \]
Here we used $1+2\log a\le a-1$. By induction and \eqref{eq:logs},
\[ N_r\le a^2\ell_r(n)\qquad(r\ge0). \]

It remains to complete the rounding of these blocks. For a nonempty block with $N$ columns, apply Proposition~\ref{prop:partial} with $b_N=C(\sqrt d+\{d\log(2N)\}^{1/3})$. Finally, applying \eqref{eq:joint-tail} to singletons and using a union bound gives
\[ \Prob(B\ne\varnothing) \le N\exp(-cb_N^3/d)\le\frac14. \]
Repeat this final application of Proposition~\ref{prop:partial}, each time from the same initial point of the block and with fresh randomness, until $B=\varnothing$. The accepted output is a complete rounding. Of course, if $b_N>c_0d$ then \eqref{eq:linear} can be used instead. Including the option of always using \eqref{eq:linear}, we can therefore finish any block with additional error at most
\begin{equation}\label{eq:terminal} T(N):=C\min\left\{d,\sqrt d+\{d\log(2N)\}^{1/3}\right\}. \end{equation}
Thus, after $r$ recursive levels followed by these final roundings,
\begin{equation}\label{eq:recurrence} \norm{A(\chi-z)}_\infty \le Crb+T(N_r). \end{equation}

For the fixed-depth bound, take $r=k$. Our size estimate gives
\[ \log(2N_k) \le\log(2a^2)+\log\ell_k(n) \le C\{\log(2ds)+\ell_{k+1}(n)\}. \]
Substituting into \eqref{eq:recurrence}, absorbing $(d\log(2d))^{1/3}$ into $O(\sqrt d)$ proves \eqref{eq:fixed}.

For the uniform bound, take $r=\lstar n$. By definition, $\ell_r(n)\le e$, so every remaining block has at most $ea^2$ columns. Since $T(ea^2)=O(b)$, \eqref{eq:recurrence} gives
\[ \norm{A(\chi-z)}_\infty\le Cb(1+\lstar n). \]
Absorbing the $\log d$ term as above and again comparing with \eqref{eq:linear} proves \eqref{eq:uniform}.

Finally, each acceptance test succeeds with probability at least $3/4$, so retries have constant expected overhead. At each depth the subproblems have disjoint column sets, and there are at most $n$ nonempty subproblems. Thus the algorithm has polynomial expected running time. The geometric tails of the retry counts also give polynomial running time with high probability.
\end{proof}

\section{The modified algorithm}\label{sec:joint}

We prove Proposition~\ref{prop:partial} by modifying the Bansal--Jiang walk~\cite[Section~5]{BJ}. We retain many of the key notions and terminology from their analysis, such as the discrepancy levels, quadratically corrected row potential, and covariance bounds. Their analysis shows that, with high probability, every column satisfies the required bounds on its weight on high-level rows. The main departure is that we modify the algorithm to allow for exceptional columns: when a column violates such a bound, we freeze its coordinate and continue the walk on the remaining coordinates. This produces a partial rounding whose exceptional coordinates can be handled by our bootstrapping scheme, provided we establish the joint estimate~\eqref{eq:joint-tail}. We emphasize that the rest of this section, while containing some new ideas about the partial coloring and handling of columns, is largely adapted from the existing analysis of \cite{BJ}.

\subsection{The covariance lemma}

\begin{lemma}[Bansal--Jiang, \cite{BJexp}, Theorem~2.12]
\label{lem:covariance}
Let $h\ge1$, and let $W\subseteq\R^h$ be a linear subspace with $\dim W\le h/10$. For each $\ell$ in a finite index set, let $u_{\ell,1},\ldots,u_{\ell,m_\ell}\in\R^h$, and suppose $\alpha_\ell>0$ satisfy $\sum_\ell m_\ell/\alpha_\ell\le h/10$. There is a mean-zero random unit vector $v\in W^\perp$ such that
\begin{align}
  \E\ip{a}{v}^2 &\le 6\sum_{j=1}^h a_j^2\E v_j^2 &&(a\in\R^h),\label{eq:coordinate-covariance}\\
  \E\left(\sum_{i=1}^{m_\ell} c_i\ip{u_{\ell,i}}{v}\right)^2 &\le\alpha_\ell\sum_{i=1}^{m_\ell} c_i^2\E\ip{u_{\ell,i}}{v}^2 &&(\text{each }\ell,\ c\in\R^{m_\ell}).\label{eq:row-covariance}
\end{align}
Its distribution can be constructed from a semidefinite program of polynomial size.
\end{lemma}
\begin{proof}
Let $E_\ell$ be the matrix with rows $u_{\ell,i}$. Ignoring empty groups, apply the cited theorem with $\delta=1/10$, $\kappa=1/6$, and $\eta=6$. Its hypothesis is satisfied since
\[ \eta^{-1}+\kappa+ \sum_\ell\frac{m_\ell}{h\alpha_\ell} \le\frac16+\frac16+\frac1{10}<1-\delta. \]
It gives a nonzero positive semidefinite matrix $U$ with $Uw=0$ for $w\in W$ and
\[ U\preceq6\operatorname{diag}(U),\qquad E_\ell U E_\ell^\top \preceq\alpha_\ell \operatorname{diag}(E_\ell U E_\ell^\top). \]
Here $\operatorname{diag}(M)$ retains only the diagonal of $M$, and $M\preceq N$ means $N-M\succeq0$. For an orthonormal eigendecomposition $U=\sum_t\rho_t e_te_t^\top$, choose $t$ with probability $\rho_t/\operatorname{tr}U$ and take $v=\pm e_t$ with equal probabilities. Then $\E vv^\top=U/\operatorname{tr}U$, proving both claims.
\end{proof}

\subsection{The walk}

In order to avoid working with absolute values, we vertically concatenate $-A$ and $A$, write $a_i$ for the resulting signed rows, and let $D=2d$. In this notation, $\sum_i|a_{ij}|\le D$ for each column $j$. It suffices to bound the increase of every signed row sum.

Let $S$ be the set of \emph{alive variables}, namely the coordinates eligible for updates during the walk. Following~\cite[Sections~3.1 and~5.1]{BJ}, say a row's \emph{size} is
\[ s_i(S):=\sum_{j\in S}|a_{ij}|. \]
The size of a row bounds its possible discrepancy increments. Since $S$ shrinks over time, sizes cannot increase.

In addition to characterizing rows based on their size, we also categorize them into ``levels'' based on their discrepancy. For a potential function $Y$ and a sequence of thresholds $b_\ell$ defined shortly, the level $\ell_i$ of row $i$ starts at zero and increases when the potential $Y_i$ reaches the next threshold $2b_{\ell_i}$ (and some size constraints, depending on $\ell$, are met). In particular, the level of a row never decreases, even if its discrepancy falls. After a row reaches level $\ell_i=L+1$, we constrain $Y_i$ to be non-increasing until enough variables have been frozen that $s_i(S) \le b$. Of course, the support-based notion of ``size'' and the discrepancy-based notion of ``level'' should interact: even if a row has large level, it is not a source for worry unless its size is significant. Correspondingly, we use level-dependent thresholds to separate medium- and large-sized rows. In our modification, we also freeze variables when their columns have too much support on the set of rows of high level.

\paragraph{Parameters.}

The algorithm parameters are:
\begin{align}
  k_\ell&=D100^{-\ell}, &L&=\min\{\ell\ge0:k_\ell\le b\},\label{eq:walk-par-1}\\
  b_\ell&=b(\ell+1)^2 5^{-\ell}, &H_\ell&=100\,2^\ell k_\ell,\label{eq:walk-par-2}\\
  \beta_\ell&=b_\ell/H_\ell, &\alpha_\ell&=A_0\,2^\ell\max\{1,k_\ell/b\}, \qquad 0\le\ell\le L.\label{eq:walk-par-3}
\end{align}

In brief, $k_\ell$ is used to bound $\sum_{i:\ell_i\ge\ell}|a_{ij}|$, i.e., how much $x_j$ affects rows that have reached discrepancy level $\ell$; $L+1$ is the final level, at which we keep $Y_i$ nonincreasing while the row size exceeds $b$; $H_\ell$ is the size threshold separating medium and large for rows with discrepancy level $\ell$, and $2b_\ell$ is the discrepancy threshold at which a medium row advances levels. The coefficient $\beta_\ell$ is used in the potential, and $\alpha_\ell$ is used in Lemma \ref{lem:covariance} to bound the variance of the increments of medium rows.

Here $A_0\ge40$ is an absolute constant, and $k_{L+1}:=k_L/100$. By adjusting the universal constants in Proposition~\ref{prop:partial}, we may assume
\begin{equation}\label{eq:walk-par-bounds} b/100<k_L\le b,\qquad k_L\ge100,\qquad \beta_\ell\le\tfrac12,\qquad \sum_{\ell=0}^L b_\ell=O(b). \end{equation}
The choices of $k_\ell$ and $H_\ell\asymp2^\ell k_\ell$ follow \cite[Section~5.1]{BJ}, while our $b_\ell$ has the additional factor $(\ell+1)^2$, making the variances summable in~\eqref{eq:all-variance}. Moreover, our $b$ is both the discrepancy target as well as the size threshold below which rows need no further constraints; two separate parameters are used in \cite{BJ} for these quantities.

\paragraph{The potential.}
As in \cite{BJ}, a row at level $\ell\le L$ is called \emph{large} if $s_i(S)>H_\ell$, and \emph{medium} if $b<s_i(S)\le H_\ell$. When a row first becomes medium at each level $\ell$, we record the value $c_i:=\ip{a_i}{x}$ and define
\begin{equation}\label{eq:Y}
\begin{split}
  Y_i&=\ip{a_i}{x}-c_i +\beta_\ell\sum_{j\in S}a_{ij}^2(1-x_j^2),\\
  u_i&=(a_{ij}-2\beta_\ell a_{ij}^2x_j)_{j\in S}.
\end{split}
\end{equation}
Thus $Y_i$ is the increase over the recorded row sum, plus a nonnegative quadratic term used in~\cite{BJ}. The latter decreases in expectation under a random move and supplies negative expected drift. When $i$ enters $\mathcal M_\ell$, we have $Y_i\le\beta_\ell H_\ell=b_\ell$. The vector $u_i$ is its gradient in the moving coordinates: expanding the squares, with $S$ fixed, gives
\begin{equation}\label{eq:delta-Y} Y_i(x+\epsilon v)-Y_i(x) =\epsilon\ip{u_i}{v} -\beta_\ell\epsilon^2\sum_{j\in S}a_{ij}^2v_j^2. \end{equation}
On reaching the final level $L+1$, $c_i$ is left unchanged (it retains the value recorded when the row entered $\mathcal M_L$), and $\beta_L$ is used in \eqref{eq:Y}. While $s_i(S)>b$, we require $\ip{u_i}{v}=0$, so $Y_i$ is nonincreasing.

\paragraph{Freezing columns.}
For every coordinate $j\in S$, we require
\begin{equation}\label{eq:column-budgets} \sum_{i:\ell_i\ge\ell}|a_{ij}|\le k_\ell \qquad(1\le\ell\le L+1). \end{equation}
If this condition fails, we freeze $x_j$ and remove $j$ from $S$ and put it into $B$. Recall that a row's level never decreases, even if its size falls below $b$, and hence remains included in the sum. We track three key sets of the rows:
\[
\begin{split}
  \mathcal M_\ell&=\{i:\ell_i=\ell,\ b<s_i(S)\le H_\ell\} \qquad(0\le\ell\le L),\\
  \mathcal H&=\{i:\ell_i\le L,\ s_i(S)>H_{\ell_i}\},\\
  \mathcal T&=\{i:\ell_i=L+1,\ s_i(S)>b\}.
\end{split}
\]
Here $\mathcal M_\ell$ is the medium rows at level $\ell$, $\mathcal H$ is the set of large rows across all levels, and $\mathcal T$ is the set of rows at level $L+1$ whose size still exceeds $b$. The walk keeps $\langle a_i,x\rangle$ constant for $i\in\mathcal H$ and keeps $Y_i$ nonincreasing for $i\in\mathcal T$.

We now give the algorithm. Whenever needed, $Y_i$ and $u_i$ are evaluated at the current $x,S$, using the recorded value $c_i$ and the coefficient $\beta_{\ell_i}$ (or $\beta_L$ when $\ell_i=L+1$).

\begin{algorithm}[H]
\caption{Partial rounding}\label{alg:partial}
\begin{algorithmic}[1]
\Require Signed rows $(a_i)_i$, initial point $z$, the parameters above, and step size $0<\epsilon<1$, chosen below.
\State Set $x=z$, $B=\varnothing$, $S=\{j:|z_j|<1-\epsilon\}$, and $\ell_i=0$ for every row.
\State Compute $\mathcal M_\ell,\mathcal H,\mathcal T$. For each $i\in\mathcal M_0$, record $c_i=\ip{a_i}{x}$.
\While{$|S|>100$}
\State Recompute $u_i$ for $i\in\mathcal T\cup\bigcup_{\ell=0}^L\mathcal M_\ell$ and set $W=\operatorname{span}\bigl(\{x_S\}\cup \{(a_i)_S:i\in\mathcal H\}\cup\{u_i:i\in\mathcal T\}\bigr)$.
\State Sample $v\in\R^S$ using Lemma~\ref{lem:covariance} with $W$, the groups $(u_i)_{i\in\mathcal M_\ell}$, and the parameters $(\alpha_\ell)_{\ell=0}^L$ from~\eqref{eq:walk-par-3}.
\State Set $x_S\gets x_S+\epsilon v$.
\State Using this new $x_S$, record $\mathcal Q=\{(i,\ell):i\in\mathcal M_\ell,\ Y_i\ge2b_\ell\}$.
\For{each $(i,\ell)\in\mathcal Q$, in fixed row order}
\State Set $\ell_i\gets\ell+1$.
\State ``Freeze'' each $x_j$ violating \eqref{eq:column-budgets} with $j\in S$ by moving $j$ from $S$ to $B$.
\EndFor
\State ``Freeze'' each $x_j$ with $j\in S$ and $|x_j|\ge1-\epsilon$ by removing $j$ from $S$ without putting it in $B$.
\State Recompute $\mathcal M_\ell,\mathcal H,\mathcal T$. For each row $i$ newly entering some $\mathcal M_\ell$, set $c_i=\ip{a_i}{x}$.
\EndWhile
\State Round the at most $100$ coordinates remaining in $S$ arbitrarily to signs.
\State Round each frozen coordinate $x_j$ with $j \not\in B$ to the nearest sign. Leave $x_j$ unchanged for $j\in B$.
\State \Return $(x,B)$.
\end{algorithmic}
\end{algorithm}

The list $\mathcal Q$ is formed before any columns are removed, and every recorded pair is processed even if earlier removals change the row's size. Column-condition violations take priority over freezing near $\pm1$.

Here, setting $W=\operatorname{span}\bigl(\{x_S\}\cup \{(a_i)_S:i\in\mathcal H\}\cup\{u_i:i\in\mathcal T\}\bigr)$ in Lemma \ref{lem:covariance} makes the statement precise about fixing row-sums for $i \in \mathcal{H}$ and forcing $Y_i$ to be non-increasing for $i \in \mathcal{T}$. The sampling of $v$, which is black-boxed as Lemma \ref{lem:covariance} here, is implemented in \cite{BJ} by solving an SDP, and is the main ``substance'' of the algorithm.

\paragraph{Applying the covariance lemma.}
We now formally justify the usage of Lemma \ref{lem:covariance} in the algorithm. At the start of an iteration of the while loop, let $h=|S|>100$ be the current number of alive variables. From \eqref{eq:column-budgets} and the original absolute column-sum bound $k_0=D$,
\[ \sum_{i:\ell_i\ge\ell}s_i(S) =\sum_{j\in S}\sum_{i:\ell_i\ge\ell}|a_{ij}| \le hk_\ell\qquad(0\le\ell\le L+1). \]
Rows in $\mathcal H$ at level $\ell$ have size greater than $H_\ell$, while rows in $\mathcal M_\ell$ and $\mathcal T$ have size greater than $b$. Dividing their respective total-size bounds by these thresholds gives
\[ \#\{i\in\mathcal H:\ell_i=\ell\}\le\frac{hk_\ell}{H_\ell}, \qquad m_\ell:=|\mathcal M_\ell|\le\frac{hk_\ell}{b}, \qquad |\mathcal T|\le\frac{hk_{L+1}}b\le\frac h{100}. \]
Consequently,
\[ \dim W\le1+\frac h{100} +\sum_{\ell=0}^L\frac h{100\,2^\ell}<\frac h{10}, \qquad \sum_{\ell=0}^L\frac{m_\ell}{\alpha_\ell} \le\frac h{A_0}\sum_{\ell=0}^L2^{-\ell}\le\frac h{10}. \]
These are precisely the two hypotheses of Lemma~\ref{lem:covariance}.

\paragraph{Termination and discrepancy.}
Each step stays in the cube because $|x_j|<1-\epsilon$ for $j\in S$ and $|v_j|\le1$. Since $v\perp x_S$ and $\|v\|_2=1$,
\begin{equation}\label{eq:norm-growth} \|x+\epsilon v\|_2^2=\|x\|_2^2+\epsilon^2, \end{equation}
where $v$ is extended by zero outside $S$. Freezing coordinates does not change $x$, so the bound $\|x\|_2^2\le n$ allows at most $n/\epsilon^2+1$ steps.

For $i\in\mathcal H$, $v\perp(a_i)_S$ keeps $\ip{a_i}{x}$ constant. For $i\in\mathcal M_\ell$, we have $|(u_i)_j|\le2|a_{ij}|$ and $\sum_{j\in S}a_{ij}^2\le H_\ell$. Thus~\eqref{eq:delta-Y} shows that $Y_i$ can exceed $2b_\ell$ by at most $2\epsilon\sqrt{H_\ell}$ in the step that increases its level. Since $\ip{a_i}{x}-c_i\le Y_i$, the increase of the signed row sum accumulated at level $\ell$ is at most $2b_\ell+2\epsilon\sqrt{H_\ell}$. For $\ell=L$, this same bound continues to hold after the row enters $\mathcal T$: we retain $c_i$ and require $v\perp u_i$, so \eqref{eq:delta-Y} makes $Y_i$ nonincreasing. Removing coordinates from $S$ also only decreases the quadratic term of $Y_i$.

Once $s_i(S)\le b$, the further change during the walk is at most $2b$. The final rounding changes each row sum by at most $200+n\epsilon$: at most $100$ coordinates are rounded arbitrarily, and every other rounded coordinate is within $\epsilon$ of its sign. Summing over levels, using~\eqref{eq:walk-par-bounds} and $\sum_\ell\sqrt{H_\ell}=O(\sqrt D)$, and considering both signed copies of each original row gives
\begin{equation}\label{eq:walk-error} \|A(x-z)\|_\infty\le Cb+O(n\epsilon). \end{equation}

\subsection{The joint tail.}\label{sec:joint tail}

It remains to prove~\eqref{eq:joint-tail}. Our concentration calculation follows the exponential-potential argument of \cite[Sections~3.2.5 and~5.2]{BJ}; see also the stopped exponential processes in \cite[Theorem~2.6]{BJexp}. As in those arguments, the quadratic correction supplies negative drift, while the covariance bound controls the variance of a sum of row exponentials.

Here we combine all levels and all columns in a fixed independent set $R$ of the column graph into one weighted potential, retaining each column's contribution after it is frozen. Our choices of $b_\ell$ and $\lambda_\ell$ make the variance bounds summable over levels, without a loss depending on the number of levels. Each frozen column in $R$ forces a constant contribution to the final potential, and distinct columns in $R$ use disjoint rows. Together with the exponential-moment estimate, this gives the factor $|R|$ in the exponent of~\eqref{eq:joint-tail}.

Let $t$ count completed iterations of the while loop, and let $x_t,S_t,\ell_i(t),c_i(t)$ denote the state after those iterations, including all level changes and column removals, so that $t=0$ is initialization. Write $\mathcal M_\ell(t),Y_i(t),u_i(t)$ for the corresponding sets and quantities from the algorithm, whenever defined, and let
\[ \tau=\min\{t:|S_t|\le100\} \]
be the number of iterations before the final rounding. For $t<\tau$, let $v_t$ be the sampled direction and $\mathcal Q_t$ the threshold list formed during the next iteration. Let $\mathcal F_t$ be the sigma-algebra generated by the first $t\wedge\tau$ iterations; for $t<\tau$, this is the history before sampling $v_t$.

Fix a nonempty independent set $R$ in the column graph and put $r=|R|$. Let $\mathcal I$ be the set of signed rows with a nonzero entry in a column of $R$. Each $i\in\mathcal I$ has exactly one such column, denoted by $j(i)$, and we write $w_i=|a_{i,j(i)}|$. For each level $\ell$, set
\[ \mathcal I_\ell(t) =\{i\in\mathcal I\cap\mathcal M_\ell(t):j(i)\in S_t\}. \]
For fixed $i$ and $\ell$, the times $\{t\,:\, i\in\mathcal I_\ell(t)\}$ form a single interval, possibly empty. Indeed, levels are non-decreasing, size is non-increasing, and columns never return to $S_t$. For $0\le\ell\le L$, define
\begin{equation}\label{eq:exp-par} \lambda_\ell=\frac{10(\ell+1)}{b_\ell},\qquad p_\ell=e^{-10(\ell+1)},\qquad \theta_\ell=\beta_\ell/24. \end{equation}

\paragraph{Auxiliary processes and the exponential potential.}

\begin{definition}[Auxiliary exponential processes]
Fix $i\in\mathcal I$ and $0\le\ell\le L$. Let
\[ \sigma=\inf\{t:0\le t\le\tau,\ i\in\mathcal I_\ell(t)\}, \qquad \rho=\inf\{t:\sigma<t\le\tau,\ i\notin\mathcal I_\ell(t)\} \]
be the entry and exit times, with the infima taken over integer times and $\inf\varnothing:=\infty$. For $t<\tau$ with $i\in\mathcal I_\ell(t)$, write
\begin{equation}\label{eq:raw-Y} \widehat Y_i(t+1) =Y_i(t)+\epsilon\ip{u_i(t)}{v_t} -\beta_\ell\epsilon^2 \sum_{j\in S_t}a_{ij}^2(v_t)_j^2, \end{equation}
Thus $\widehat Y_i(t+1)$ is the value of $Y_i$ immediately after applying the random move, before changing the level of $i$, removing columns, or resetting $c_i$. For $0\le t\le\tau$, define
\begin{equation}\label{eq:process-update}
  X_i^{(\ell)}(t)=
\begin{cases}
  1, &t<\sigma,\\[2pt]
  \exp\{\lambda_\ell(Y_i(t)-b_\ell)\}, &\sigma\le t<\rho,\\[2pt]
  \min\!\left\{p_\ell^{-1}, \exp\{\lambda_\ell(\widehat Y_i(\rho)-b_\ell)\}\right\}, &\rho\le t.
\end{cases}
\end{equation}
For $t>\tau$, set $X_i^{(\ell)}(t)=X_i^{(\ell)}(\tau)$. For $0\le t<\tau$, also define
\[
  \widehat X_i^{(\ell)}(t+1)=
\begin{cases}
  1, &t<\sigma,\\[2pt]
  \exp\{\lambda_\ell(\widehat Y_i(t+1)-b_\ell)\}, &\sigma\le t<\rho,\\[2pt]
  X_i^{(\ell)}(\rho), &\rho\le t.
\end{cases}
\]
Thus $\widehat X_i^{(\ell)}(t+1)$ uses only the random move, before capping, level changes, column removals, and assignments of new $c_i$. The final rounding is not part of either process.
\end{definition}

Two properties of this definition will be used below:
\begin{equation}\label{eq:process-bounds} 0<X_i^{(\ell)}(t)\le p_\ell^{-1},\qquad X_i^{(\ell)}(t+1)\le\widehat X_i^{(\ell)}(t+1) \quad(t<\tau). \end{equation}
With these auxiliary processes constructed, we are ready to state the exponential potential. Define:
\begin{equation}\label{eq:potential} \Psi(t)=\sum_{\ell=0}^L\Psi_\ell(t),\qquad \Psi_\ell(t)=\frac{p_\ell}{k_\ell} \sum_{i\in\mathcal I}w_iX_i^{(\ell)}(t). \end{equation}
Write $\widehat\Psi_\ell(t+1)$ and $\widehat\Psi(t+1)$ for the same sums with $X$ replaced by $\widehat X$. Since $\sum_{i\in\mathcal I}w_i\le rD$ and $X_i^{(\ell)}(0)\le1$,
\begin{equation}\label{eq:potential-initial} \Psi(0)\le r\sum_{\ell=0}^L\frac{D}{k_\ell}p_\ell \le\frac{re^{-10}}{1-100e^{-10}}<\frac r{1000}. \end{equation}
Then \eqref{eq:process-bounds} gives $\Psi(t+1)\le\widehat\Psi(t+1)$. Note that in $\Psi_{\ell}(t)$, only the terms indexed by $\mathcal I_\ell(t)$ change during increments of the random walk. The column bounds~\eqref{eq:column-budgets}, together with the original absolute column-sum bound $k_0=D$ for $\ell=0$, give
\begin{equation}\label{eq:tracked-row-sum} \sum_{i\in\mathcal I_\ell(t)}w_i \le\sum_{j\in R\cap S_t}\sum_{i:\ell_i(t)\ge\ell}|a_{ij}| \le rk_\ell, \qquad \text{ hence } \qquad \frac{p_\ell}{k_\ell} \sum_{i\in\mathcal I_\ell(t)}w_iX_i^{(\ell)}(t)\le r. \end{equation}
Thus the total contribution to $\Psi_\ell(t)$ from terms that evolve during the random move is at most $r$.

\paragraph{Exponential estimate.}
Our goal is to show that the potential does not grow too large. We do this by bounding its exponential moment, which will give the desired tail bound by Markov's inequality.

\begin{lemma}
For a sufficiently small absolute constant $c_1>0$ and a sufficiently large absolute constant $C_1$, setting $\nu=c_1b^3/D$ and $\epsilon=(nd)^{-C_1}$ gives
\begin{equation}\label{eq:potential-final} \E e^{\nu\Psi(\tau)}\le\exp\{\nu r/500\}. \end{equation}
\end{lemma}

\begin{proof}
Let $Q_r=1+(1+\nu)r(L+1)$. The majority of the proof consists of showing the one-step estimate
\begin{equation}\label{eq:multiplicative} \frac{\E[e^{\nu\Psi(t+1)}\mid\mathcal F_t]} {e^{\nu\Psi(t)}} \le1+CQ_r^3\epsilon^3 \qquad(t<\tau), \end{equation}
provided $\epsilon Q_r$ is sufficiently small. We bound the ratio using the raw update. In the conditional Taylor expansion of this upper bound, the coefficient of $\epsilon$ is zero, and we will show that the coefficient of $\epsilon^2$ is nonpositive. We compute the expansion in several steps, first bounding the increments of $X_i$, then $\Psi_\ell$, then finally $\Psi$ and the exponential. When working with the increments of $X_i$, it will also be convenient to use the ``raw'' update $\widehat X_i^{(\ell)}(t+1)-X_i^{(\ell)}(t)$, which suffices for an upper bound in view of \eqref{eq:process-bounds}.

We begin by estimating the change caused by an increment of the random walk, before changing levels or freezing columns. We first expand $\widehat X_i^{(\ell)}(t+1)-X_i^{(\ell)}(t)$ and then sum to obtain $\widehat\Psi(t+1)-\Psi(t)$. By~\eqref{eq:process-bounds}, the latter is an upper bound on the actual potential increment.

Fix $t<\tau$, $0\le\ell\le L$, and $i\in\mathcal I_\ell(t)$, and condition on $\mathcal F_t$. The current state, including $S_t$, $u_i(t)$, and $X_i^{(\ell)}(t)$, is then fixed. By the definitions of $X_i^{(\ell)}$ and $\widehat X_i^{(\ell)}$, a second-order Taylor expansion gives
\begingroup
\begin{equation}\label{eq:row-expansion}
\begin{aligned}
  \frac{\widehat X_i^{(\ell)}(t+1)}{X_i^{(\ell)}(t)} &=\exp\left\{ \epsilon\lambda_\ell\ip{u_i(t)}{v_t} -\epsilon^2\lambda_\ell\beta_\ell \sum_{j\in S_t}a_{ij}^2(v_t)_j^2 \right\}\\
  &=1+\epsilon\lambda_\ell\ip{u_i(t)}{v_t} +\epsilon^2\left( \frac{\lambda_\ell^2}{2}\ip{u_i(t)}{v_t}^2 -\lambda_\ell\beta_\ell \sum_{j\in S_t}a_{ij}^2(v_t)_j^2 \right) +O(\epsilon^3).
\end{aligned}
\end{equation}
\endgroup
We first justify that the remainder is uniform. Since $|x_{t,j}|\le1$, $|a_{ij}|\le1$, and $\beta_\ell\le1/2$,
\[ |(u_i(t))_j| =|a_{ij}-2\beta_\ell a_{ij}^2x_{t,j}| \le(1+2\beta_\ell)|a_{ij}| \le2|a_{ij}|. \]
Also, $\|v_t\|_2=1$ and $\sum_{j\in S_t}a_{ij}^2\le H_\ell$, because $i\in\mathcal M_\ell(t)$. Together with the parameter choices, these give
\begin{equation}\label{eq:increment-bounds} \begin{aligned} \lambda_\ell|\ip{u_i(t)}{v_t}| \le2\lambda_\ell\sqrt{H_\ell}\le C, \qquad 0\le\lambda_\ell\beta_\ell \sum_{j\in S_t}a_{ij}^2(v_t)_j^2 &\le\lambda_\ell\beta_\ell\le C. \end{aligned} \end{equation}
For $0<\epsilon\le1$, the Taylor remainder is therefore bounded in absolute value by $C'\epsilon^3$, uniformly over rows, levels, histories, and sampled directions.

We now take conditional expectations. The term of order $\epsilon$ vanishes because $\E[v_t\mid\mathcal F_t]=0$. It remains to compare the two competing terms of order $\epsilon^2$: the positive term comes from exponentiating the mean-zero linear increment, while the negative term comes from the quadratic correction in $Y_i$. Using~\eqref{eq:coordinate-covariance} and the bound on $|(u_i(t))_j|$ gives
\[ \begin{aligned} \E[\ip{u_i(t)}{v_t}^2\mid\mathcal F_t] &\le6\sum_{j\in S_t}(u_i(t))_j^2 \E[(v_t)_j^2\mid\mathcal F_t] \le24\sum_{j\in S_t}a_{ij}^2 \E[(v_t)_j^2\mid\mathcal F_t]. \end{aligned} \]
Consequently, recalling $\theta_\ell=\beta_\ell/24$,
\[
\begin{aligned}
  \frac{ \E[\widehat X_i^{(\ell)}(t+1) -X_i^{(\ell)}(t)\mid\mathcal F_t] }{X_i^{(\ell)}(t)} &=\epsilon^2\left( \frac{\lambda_\ell^2}{2} \E[\ip{u_i(t)}{v_t}^2\mid\mathcal F_t] -\lambda_\ell\beta_\ell \sum_{j\in S_t}a_{ij}^2 \E[(v_t)_j^2\mid\mathcal F_t] \right)+O(\epsilon^3)\\
  &\quad\le\epsilon^2 \left(\frac{\lambda_\ell^2}{2} -\lambda_\ell\theta_\ell\right) \E[\ip{u_i(t)}{v_t}^2\mid\mathcal F_t] +O(\epsilon^3).
\end{aligned}
\]
Our parameter choices and $b\ge C\sqrt d$ give
\begin{equation}\label{eq:lambda-theta} \frac{\lambda_\ell}{\theta_\ell} =\frac{24000D}{b^2(\ell+1)^3 2^\ell}\le1. \end{equation}
Thus the positive second-order term has magnitude at most half of the negative term, and
\begin{equation}\label{eq:row-drift} \E[\widehat X_i^{(\ell)}(t+1) -X_i^{(\ell)}(t)\mid\mathcal F_t] \le -\frac{\lambda_\ell\theta_\ell}{2} X_i^{(\ell)}(t) \, \E[\ip{u_i(t)}{v_t}^2\mid\mathcal F_t]\,\epsilon^2 +O\!\left(X_i^{(\ell)}(t)\epsilon^3\right). \end{equation}

This provides the desired row-wise control on raw increments. We now turn to the increments of $\Psi_\ell$. Subtracting one from~\eqref{eq:row-expansion}, multiplying by $p_\ell w_iX_i^{(\ell)}(t)/k_\ell$, and summing over $i\in\mathcal I_\ell(t)$ yields
\begingroup
\begin{equation}\label{eq:potential-expansion}
\begin{aligned}
  \widehat\Psi_\ell(t+1)-\Psi_\ell(t) &=\epsilon\,\frac{p_\ell\lambda_\ell}{k_\ell} \sum_{i\in\mathcal I_\ell(t)} w_iX_i^{(\ell)}(t)\ip{u_i(t)}{v_t}\\
  &\quad+\epsilon^2\,\frac{p_\ell}{k_\ell} \sum_{i\in\mathcal I_\ell(t)} w_iX_i^{(\ell)}(t) \left( \frac{\lambda_\ell^2}{2}\ip{u_i(t)}{v_t}^2 -\lambda_\ell\beta_\ell \sum_{j\in S_t}a_{ij}^2(v_t)_j^2 \right) +O(r\epsilon^3).
\end{aligned}
\end{equation}
\endgroup
The remainder is again uniform  by~\eqref{eq:tracked-row-sum}. Define
\begin{equation}\label{eq:drift}
\begin{aligned}
  Z_\ell(t) &=\frac{p_\ell\lambda_\ell}{k_\ell} \sum_{i\in\mathcal I_\ell(t)} w_iX_i^{(\ell)}(t)\ip{u_i(t)}{v_t},\\
  \mathcal D_\ell(t) &=\frac{p_\ell\lambda_\ell\theta_\ell}{2k_\ell} \sum_{i\in\mathcal I_\ell(t)} w_iX_i^{(\ell)}(t) \E[\ip{u_i(t)}{v_t}^2\mid\mathcal F_t].
\end{aligned}
\end{equation}
In~\eqref{eq:potential-expansion}, $Z_\ell(t)$ is the coefficient of $\epsilon$ and has conditional mean zero. The quantity $-\mathcal D_\ell(t)$ bounds the conditional expectation of the coefficient of $\epsilon^2$, as is seen by summing \eqref{eq:row-drift} over $i\in\mathcal I_\ell(t)$ with weights $p_\ell w_i/k_\ell$. Thus, recalling~\eqref{eq:tracked-row-sum},
\begingroup
\begin{equation}\label{eq:potential-drift} \E[\widehat\Psi_\ell(t+1)-\Psi_\ell(t)\mid\mathcal F_t] \le-\epsilon^2\mathcal D_\ell(t)+Cr\epsilon^3. \end{equation}
\endgroup

We now bound the variance of $Z_\ell(t)$, which will control the positive second-order term arising when we Taylor expand the exponential moment of $\Psi$. As $Z_\ell(t)$ is conditionally centered, its conditional second moment is equal to its conditional variance. Applying~\eqref{eq:row-covariance} with coefficients $w_iX_i^{(\ell)}(t)$ for $i\in\mathcal I_\ell(t)$ and zero coefficients on $\mathcal M_\ell(t)\setminus\mathcal I_\ell(t)$, we claim:
\begin{align}
  \E[Z_\ell(t)^2\mid\mathcal F_t] &\le \left(\frac{p_\ell\lambda_\ell}{k_\ell}\right)^2 \alpha_\ell \sum_{i\in\mathcal I_\ell(t)} w_i^2\bigl(X_i^{(\ell)}(t)\bigr)^2 \E[\ip{u_i(t)}{v_t}^2\mid\mathcal F_t]\notag\\
  &\le \frac{2\alpha_\ell\lambda_\ell}{\theta_\ell k_\ell} \mathcal D_\ell(t) \le\frac{CD}{b^3(\ell+1)^3}\mathcal D_\ell(t). \label{eq:level-variance}
\end{align}
Some justification is in order. For the second inequality, we used
\[ \bigl(w_iX_i^{(\ell)}(t)\bigr)^2 \le p_\ell^{-1}w_iX_i^{(\ell)}(t), \]
which follows from $w_i\le1$ and $X_i^{(\ell)}(t)\le p_\ell^{-1}$, and then substituted the definition of $\mathcal D_\ell(t)$. The last inequality uses~\eqref{eq:lambda-theta} and
\[ \frac{\alpha_\ell}{k_\ell} =A_0\,2^\ell\max\{1/k_\ell,1/b\} \le\frac{100A_0\,2^\ell}{b}, \]
since $k_\ell\ge k_L>b/100$.

With this control established over the increments of $\Psi_\ell$, it remains to sum over $\ell$. Adopting the shorthand $ Z(t)=\sum_{\ell=0}^L Z_\ell(t)$ and $ \mathcal D(t)=\sum_{\ell=0}^L\mathcal D_\ell(t)$, combining the increment estimates over the levels $\ell$ gives
\[ \E[\widehat\Psi(t+1)-\Psi(t)\mid\mathcal F_t] \le-\epsilon^2\mathcal D(t)+Cr(L+1)\epsilon^3. \]
Moreover,
\begin{equation}\label{eq:all-variance} \begin{aligned} \E[Z(t)^2\mid\mathcal F_t] &\le \left(\sum_{\ell=0}^L \sqrt{\E[Z_\ell(t)^2\mid\mathcal F_t]}\right)^2 \le\frac{CD}{b^3} \left(\sum_{\ell=0}^L \frac{\sqrt{\mathcal D_\ell(t)}}{(\ell+1)^{3/2}} \right)^2 \le\frac{CD}{b^3}\mathcal D(t). \end{aligned} \end{equation}
The last inequality follows from Cauchy--Schwarz, using $\sum_{\ell\ge0}(\ell+1)^{-3}<\infty$.

We can now prove the one-step exponential estimate \eqref{eq:multiplicative}. Summing~\eqref{eq:potential-expansion} over $0\le\ell\le L$ and using~\eqref{eq:tracked-row-sum} and~\eqref{eq:increment-bounds} gives the pointwise bounds
\[
\begin{aligned}
  |\widehat\Psi(t+1)-\Psi(t)| &\le Cr(L+1)\epsilon,\\
  |\widehat\Psi(t+1)-\Psi(t)-\epsilon Z(t)| &\le Cr(L+1)\epsilon^2.
\end{aligned}
\]
In particular,
\[ \bigl(\widehat\Psi(t+1)-\Psi(t)\bigr)^2 =\epsilon^2Z(t)^2 +O\!\left(r^2(L+1)^2\epsilon^3\right). \]
Thus, provided $\epsilon Q_r$ is sufficiently small, Taylor expansion of the outer exponential gives
\begin{align*}
  \frac{\E[e^{\nu\Psi(t+1)}\mid\mathcal F_t]} {e^{\nu\Psi(t)}} &\le \E\!\left[ e^{\nu(\widehat\Psi(t+1)-\Psi(t))} \,\middle|\,\mathcal F_t \right]\\
  &=1+\nu\,\E[\widehat\Psi(t+1)-\Psi(t)\mid\mathcal F_t] +\frac{\nu^2\epsilon^2}{2} \E[Z(t)^2\mid\mathcal F_t] +O(Q_r^3\epsilon^3)\\
  &\le1+\left( -\nu\mathcal D(t) +\frac{\nu^2}{2}\E[Z(t)^2\mid\mathcal F_t] \right)\epsilon^2 +CQ_r^3\epsilon^3\\
  &\le1+CQ_r^3\epsilon^3.
\end{align*}
The first inequality uses $\Psi(t+1)\le\widehat\Psi(t+1)$. The remainder is uniform by the preceding pointwise bounds. The last inequality follows from~\eqref{eq:all-variance} and $\nu=c_1b^3/D$, with $c_1$ sufficiently small. This proves~\eqref{eq:multiplicative}.

With \eqref{eq:multiplicative} established, it remains to choose the step size and iterate. Take $\epsilon=(nd)^{-C_1}$, independently of $R$, with $C_1$ sufficiently large that $\epsilon Q_n$ is sufficiently small and
\[ CnQ_n^3\epsilon\le\nu/1000, \qquad n\epsilon\le b. \]
Such a choice is possible since $L=O(\log(2d))$ and $1\lesssim\nu\lesssim d^2$. The process $\Psi(t)$ and the history are constant after $\tau\le n/\epsilon^2+1$, so the conditional ratio in~\eqref{eq:multiplicative} equals $1$ on $\{\tau\le t\}$. We may therefore iterate that inequality up to the deterministic bound $\lceil n/\epsilon^2\rceil+1$. Using \eqref{eq:potential-initial}, $Q_r\le Q_n$, and the choice of $\epsilon$, and absorbing an absolute factor into $C$, gives
\[ \E e^{\nu\Psi(\tau)} \le\exp\{\nu r/1000+\nu/1000\} \le\exp\{\nu r/500\}. \]
\end{proof}

\paragraph{Conclusion.}
We now show that $\{R\subseteq B\} \implies \{\Psi(\tau)>r/100\}$. The exponential-moment bound~\eqref{eq:potential-final} will then give the desired probability estimate.

Fix $j\in R\cap B$, and let $t<\tau$ be such that the algorithm moves $j$ from $S_t$ to $B$ during the next iteration. By the freezing rule~\eqref{eq:column-budgets}, this happens because the sum of absolute entries in column $j$ over rows that have reached some level $\ell+1$ exceeds $k_{\ell+1}$, for some $0\le\ell\le L$. Since row levels never decrease, the corresponding sum still exceeds this threshold at the end of the iteration:
\[ \sum_{i:\ell_i(t+1)\ge\ell+1}|a_{ij}| >k_{\ell+1}=k_\ell/100. \]

Each row contributing to this sum crossed its level-$\ell$ threshold during some iteration from time $s$ to $s+1$, where $s\le t$. At the start of that iteration, the row belonged to $\mathcal M_\ell(s)$ and $j$ was still alive, so $i\in\mathcal I_\ell(s)$. Its crossing was recorded in $\mathcal Q_s$, with $\widehat Y_i(s+1)\ge2b_\ell$. The row's level then increased, and the exit case of~\eqref{eq:process-update} gives
\[ X_i^{(\ell)}(\tau)=p_\ell^{-1}. \]
This also holds when the crossing and the freezing of $j$ occur in the same iteration: the list $\mathcal Q_s$ is formed before any columns are removed. Consequently, the terms associated with column $j$ in $\Psi_\ell(\tau)$ satisfy
\[ \frac{p_\ell}{k_\ell} \sum_{i\in\mathcal I:j(i)=j} w_iX_i^{(\ell)}(\tau) \ge \frac{1}{k_\ell} \sum_{i:\ell_i(t+1)\ge\ell+1}|a_{ij}| >\frac1{100}. \]
Because $R$ is independent in the column graph, no signed row meets two columns of $R$. Thus different columns use disjoint terms of the potential. On the event $R\subseteq B$, each of the $r$ columns therefore contributes more than $1/100$ to $\Psi(\tau)$, even if the relevant level $\ell$ differs between columns. Hence, as claimed, $\{R\subseteq B\} \implies \{\Psi(\tau)>r/100\}$. Markov's inequality and~\eqref{eq:potential-final} now give
\[ \begin{aligned} \Prob(R\subseteq B) \le \Prob(\Psi(\tau)>r/100) \le e^{-\nu r/100}\,\E[ e^{\nu\Psi(\tau)}] \le e^{-\nu r(1/100-1/500)} \le e^{-\frac{cb^3r}d}. \end{aligned} \]
Combining this with~\eqref{eq:walk-error} and noting $n\epsilon\le b$ by our parameter choices, Proposition~\ref{prop:partial} follows. \qed

\end{document}